\documentclass[a4paper,12pt]{article}
\usepackage{url}
\usepackage[dvips]{color}    
\usepackage{subfigure}
\usepackage{float,latexsym,amsmath,amsthm,amssymb,lineno}

\newtheorem{lemma}{Lemma}

\usepackage{tabularx}
\usepackage{xcolor}

\newtheorem{lem}{Lemma}[section]
\newtheorem{thm}{Theorem}[section]

 \usepackage[toc,page]{appendix} 
 \usepackage{mathtools}
\usepackage{amssymb}
 \usepackage{float}

\usepackage{amsfonts,amsmath}
  \usepackage[utf8]{inputenc}
\usepackage{graphicx}

\usepackage{caption}

\usepackage{lipsum}

  \usepackage[T1]{fontenc}
 \usepackage{lmodern}
 \usepackage{amsmath,amsfonts}
\usepackage{color}
\usepackage{amsthm}

\theoremstyle{plain}
\newtheorem{Theorem}{Theorem}

\title{Mathematical Modeling and Optimal Control of  Untrue Information : Dynamic SEIZ in Online Social Networks }

\begin{document}
\maketitle

\centerline{\scshape Fulgence Mansal }
\medskip
{\footnotesize

 \centerline{Universit\'e Catholique de l'Afrique de l'Ouest/ UUZ}
   \centerline{ Laboratory Decision Mathematics and Numerical Analysis (LMDAN / UCAD)}

} 

\medskip

\centerline{\scshape  Ibrahima Faye}
\medskip
{\footnotesize

 \centerline{Universit\'e Alioune DIOP de Bambey, Senegal}
   \centerline{Laboratory Decision Mathematics and Numerical Analysis  (LMDAN)}

}

\maketitle 
\begin{abstract}
We propose to model the phenomenon of the spread of a rumor in this paper. We manipulate a model that is based on SEIR model that specializes in spreading rumors. In the second part, we introduce a control strategy to fight against the diffusion of the rumor. Our main objective is to characterize the three optimal controls that minimize the number of spreaders, susceptibles who enter and spread the rumor, and skeptics. For that matter, using the maximum principle of Pontryagin, we prove the existence and give characterization of our controls. To illustrate the theoretical results obtained, numerical simulations are given to concretize our approach.
\end{abstract}

\noindent

\section{Introduction}

%\section{Related Works}
The phenomenon of rumor is a complex phenomenon that has eluded man since ancient times, where it intersects many factors and interventions, including what is natural, sociological, economic, and psychological. Communities have known over the years the emergence of many rumors that have spread widely among them; it was also the focus of interaction and analysis by the commanders of these societies throughout history \cite{Allport}; human beings have fabricated rumors and disseminated them for political, economic, and social purposes \cite{Bordia} , where they are exploited to achieve commercial profits or to achieve victories in wars by dissolving fear and surrender within the enemy or with holding confidence in their leaders. The phenomenon of rumor has known many changes in its composition, in line with the change that societies know and the development of daily life in general with the increasing use of technological instruments and modern technologies in communication within communities. This phenomenon has witnessed a dramatic rise and an increase in the speed of its spread. This increase contributes significantly to huge consequences on the other hand. The development of the phenomenon of rumors and the strength of their influence and impact within societies gave this phenomenon another dimension \cite{Ghazzali} , as it became used by the media and intelligence in competition between countries and what is known as propaganda and polemic or buzz by publishing some false news in whole or in part to influence the opinions of voters by raising or decreasing the popularity of politicians as happened in the elections between Trump and Hillary where Hillary was the most popular and was the favorite to win until the last weeks before the presidential election \cite{Sabato} , where some of the specialized communication agencies published many news about Hillary contributed significantly to influence public opinion tendency to Trump, who eventually won. 

Mathematical modeling is one of the most important applications of mathematics that contribute to the representation and simulation of social, economic, biological, and ecological phenomena and convert them into mathematical equations that are formulated, studied, analyzed, and interpreted see \cite{Maki}. In this context, many researchers have developed different mathematical models representing the dynamics of the rumor \cite{Rapoport} . 

 In the work \cite{Ndii} , authors gave a review and a study of several mathematical models of rumor’s propagation.

Related Work. In 1964, Goffman and Newill developed in  \cite{Goffman}  a new concept for modeling the transmission of ideas within a society based on the mathematical model $SIR$ due to the great similarity between the two phenomena. 

With the development of societies and the emergence of modern technological means (transport communication), new factors have emerged that further complicate the phenomenon of rumor and contribute to the large spread of rumors; 

As an example, in the work done by Luıs M.A. Bettencourt et al. \cite{Daley}  , the authors  proposed a new model taking into account new factors by extending the $SIR$ model to a $SEIZR$ model with two additional compartments.

With the emergence of social networks and their impacts on communication within communities where they are taking more and more space within the community, it became clear that they must be taken into account as major intervening in the spread of rumors; in this context many of the works that adopted this hypothesis have been produced. 

To reduce the negative effect of rumor propagation, in this paper, we introduce a compartmental model of rumor propagation, which considers the rumor refutation of public and information feedback  \cite{Chen} . 

Compartmental models are a mathematical approach applied to measure and predict the spread of various infectious diseases. The method of misinformation diffusion is usually a similar approach as a virus spreading process. In transmission epidemics, there is each user infected with viruses and can become susceptible to viruses. In \cite{Bettencourt}, the authors proposed the $SEIZ$ model where the skeptics are the individuals becoming immune to infection. Although it is similar to the removed (R) individuals, skeptic transitions directly from the susceptible state and their interaction will still affect different compartments as well. In this way in \cite{Isea} authors proposed a model where the rumor spreads between two different scenarios and which do not share information with each other.

In order to demonstrate the effectiveness of the model we have proposed, we will present a numerical simulation with the following figure so that we can see how well the model adapts to reality. Initial values are approximate data that we suggested after studying and researching some statistics about the users of social networks; the values are attached in the table.

The key contributions of this paper are:  We demonstrate the capability of the  $SEIZ$ model to quantify compartment transition dynamics. We showcase how such information could facilitate the development of screening criteria for distinguishing rumors from real news happenings. %on Twitter.
\\

The paper is organized as follows. In section 2 is given the model formulation.
In section 3, we give some basic models on the model.
Section 4 is devoted to optimal control problem and in the last section, numericals simulations are given. 
In Section 5, we give the concluding remarks.

\section{ Model formulation}
Compartmental models are a mathematical approach used to evaluate and predict the spread of various infectious diseases \cite{Estabrook}. At the beginning, mathematical models for the rumors were considered merely speculative and imprecise, but for the fact that rumor spreading is now seen like the transmission of disease \cite{Kermack}. The spreading of rumor is in many ways similar to the spreading of epidemic infection by the spreader or the infections to notify or infect the susceptible \cite{Anderson}. $SEIZ$  model is a compartmental model that breaks the population into distinct compartments and establishing parameters for the rates at which the population transitions between compartments. These parameters are obtained by looking at the relationships between each class of the population and making assumptions about the disease.

One drawback of the SIS model is that once a susceptible individual gets exposed to disease, he can only directly transition to infected status. In fact, especially on Twitter, this assumption does not work well; people’s ideologies are complex and when they are exposed to news or rumors, they may hold different views, take time to adopt an idea, or even be skeptical to some facts. In this situation, they might be persuaded to propagate a story, or commence only after careful consideration themselves. Additionally, it is quite conceivable that an individual can be exposed to a story (i.e. received a tweet), yet never post a tweet themselves.

Based on this reasoning, we considered a more applicable, robust model, the $SEIZ$ model which was first used to study the adoption of Feynman diagrams . In the context of Twitter, the different compartments of the $SEIZ$ model can be viewed as follows: Susceptible (S) represents a user who has not heard about the news yet; infected (I) denotes a user who has tweeted about the news; skeptic (Z) is a user who has heard about the news but chooses not to tweet about it; and exposed (E) represents a user who has received the news via a tweet but has taken some time, an exposure delay, prior to posting. We note that referring to the Z compartment as skeptics is in no way an implication of belief or skepticism of a news story or rumor. We adopt this terminology as this was the nomenclature used by the original authors of the $SEIZ$ model.

A major improvement of the $SEIZ$ model over the SIS model is the incorporation of exposure delay. That is, an individual may be exposed to a story, but not instantaneously tweet about it. After a period of time, he may believe it and then be promoted to the infected compartment. Further, it is now possible for an individual in this model to receive a tweet, and not tweet about it themselves. As shown in Figure (\ref{pred1}) , $SEIZ$ rules can be summarized as follows:
\begin{enumerate}
\item Skeptics recruit from the susceptible compartment with rate b, but these actions may result either in turning the individual into another skeptic (with probability $l$), or it may have the unintended consequence of sending that person into the exposed (E) compartment with probability $(1 - l)$.
\item A susceptible individual will immediately believe a news story or rumor with probability $p$, or that person will move to the exposed (E) compartment with probability $(1 - p)$.
\item Transitioning of individuals from the exposed compartment to the infected class can be caused by one of two separate mechanisms: 
\begin{itemize}
\item recruitment into the susceptible  compartment is in constant rate,
\item an individual in the exposed class has further contact with an infected individual (with contact rate $\rho$), and this additional contact promotes him to infected; 
\item an individual in the exposed class may become infected purely by self- adoption (with rate $\varepsilon$), and not from additional contact with those already infected.
\end{itemize}
\end{enumerate}

The $SEIZ$ model is mathematically represented by the following system of ODEs. A slight difference of our implementation of this model is that we do not incorporate vital dynamics, which includes the rate at which individuals enter and leave the population N. In epidemiological disease applications, this encompasses the rate at which people become susceptible (e.g. born) and deceased. In our application, a  topic has a net duration not exceeding several days. Thus, the net entrance and exodus of  users over these relatively short time periods is not expected to noticeably impact compartment sizes and our ultimate findings.

 From this extension, the $SEIZ$ model explored one more compartment Skeptics $(Z)$. 

In this model, the susceptibility immediately infected with probability $p$, and $(1-p)$ is the possibility of an individual transiting to the incubator class instead, from which they adopted. 

After  the contact of an infected and a Skeptic, the Skeptic succeeds to convince him that the information is false at a rate $\lambda Z I $ ; after a certain period, a portion of the infected decide not to spread the rumor at a rate $\delta I$.

$N(t)$ denotes the total population where the network has a disease-free status with $S^{*}=N, E^{*}=I^{*}=Z^{*}=0.$ Fig. 4 illustrates the relationship between each compartment.

\begin{figure}[H]
	\centering
		\includegraphics[width=0.6\textwidth]{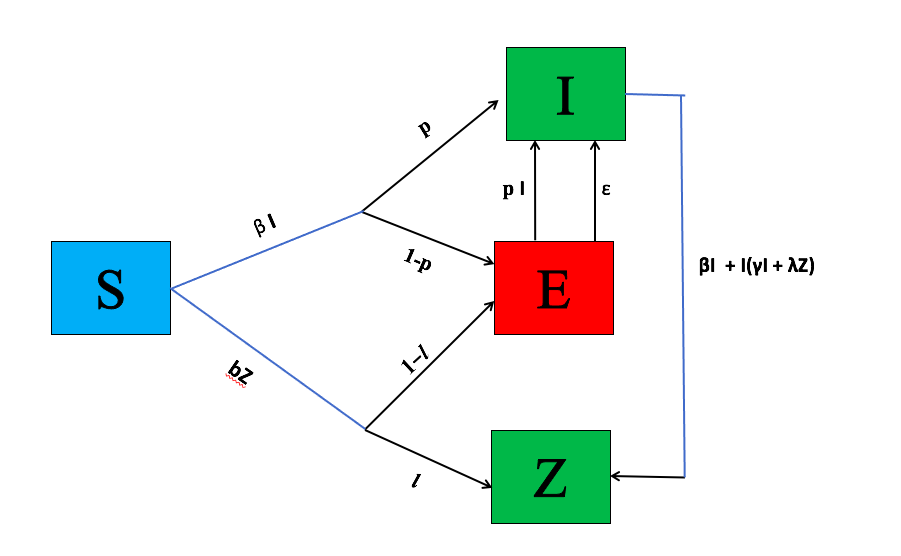}
	\caption{Transition rates of SEIZ Model.}
	\label{pred1}
 \end{figure} 
 
With the relationships between each compartment described by the parameters above, we have the following set of ODEs:

\begin{equation}
\left\{
\begin{array}{rl}
\dfrac{dS}{dt} & = \pi N- \mu S - \beta S \dfrac{I}{N} - b S\dfrac{Z}{N}  \\ \\
\dfrac{dE}{dt} & = (1-p)\beta S  \dfrac{I}{N} +  (1-l) b S\dfrac{Z}{N} - \rho E \dfrac{I}{N} - \varepsilon E - \mu E \\ \\
\dfrac{dI}{dt} & =p\beta S \dfrac{I}{N}  + \rho E \dfrac{I}{N} + \varepsilon E  - \delta I -    \lambda I \dfrac{Z}{N} - \mu I , \\ \\
\dfrac{dZ}{dt} & =lbS \dfrac{Z}{N} + \delta I +   \lambda I  \dfrac{Z}{N}  - \mu Z
\end{array}
\right.
\label{mainsys}
\end{equation}

In order to express system of equation (\ref{mainsys}) as a portion of the entire population, and since the recovered class does not appear in the first four equations of the system (\ref{mainsys}), we use the following substitution:

\begin{equation*}
s=  \dfrac{S}{N};\  i=  \dfrac{I}{N}; \ e=  \dfrac{E}{N}; \ z= \dfrac{Z}{N}
\label{mainsys3}
\end{equation*} 

is used. \\
Hence, the resulting system of equation shall be :

\begin{equation}
\left\{
\begin{array}{rl}
\dfrac{ds}{dt} & = \pi - \mu s - \beta si - b sz \\ \\
\dfrac{de}{dt} & = (1-p)\beta si +  (1-l) b sz - \rho ei - \varepsilon e - \mu e \\ \\
\dfrac{di}{dt} & =p\beta si  + \rho ei + \varepsilon e  - \delta i  -   \lambda i z - \mu i \\ \\
\dfrac{dz}{dt} & =lbsz + \delta i  +  \lambda iz - \mu z
\end{array}
\right.
\label{mainsys12}
\end{equation}

Table \ref{tab1}  provide description of each parameter of this model. This provides a more intuitive look into the model and relates the relationships above with the actual equations. The definitions of the parameters are consigned in Table \ref{tab1}. 
\begin{table}[H]
	\caption{Parameters model formulations and their descriptions}
	\begin{center}
		\begin{tabular}{c|c|c|c|c|c|c|c|c|cl} 
			\begin{tabularx}{15cm}{|c|X|c}
				\hline
				Parameter & Description   &  Units \\
				\hline
				$\pi $ & Recruitment rate into the Susceptible population& Per unit time \\
				\hline
				$\beta $ & Rate of contact between S and I  & Per unit time \\
				\hline
				$b $ & Rate of contact between S and Z & Per unit time \\
				\hline
				$\rho$ & Rate of contact between E and I  & Per unit time  \\
				\hline
				$\varepsilon$ & Incubation rate  & Per unit time  \\
				\hline
				$\dfrac{1}{ \varepsilon}$ & Average Incubation rate  & Per unit time  \\  
				\hline
				$p $ & Transmission rate S->I, given contact with I & Unit-less\\
				\hline
				$l $ & Transmission rate S->Z, given contact with Z & Unit-less\\
				\hline
				$1-l $ & S->E Probability given contact with skeptics  & Unit-less   \\
				\hline
				$1-p $ & S->E Probability given contact with adopters & Unit-less  \\
				\hline
				$\mu $ & Deconnect rate of network & Per unit time\\
				\hline
			\end{tabularx}
		\end{tabular}
	\end{center}
	\label{tab1}
\end{table}%

\section{Model Basic Properties}

Since the model monitors human populations, all the variables and the associated parameters are non-negative at all time. It is important to show that the model variables of the model remain non negative for all non-negative initial conditions.

\subsection{Positivity of the solution}
Since the model monitors population for a different class, it is required to show that all the state variables remain nonnegative for all times.
\begin{Theorem}
Let $\Omega=\left\{ (s,e,i,z) \in \mathbb{R}^4 : s(0) > 0 , e(0) > 0 , i(0) > 0, z(0) > 0 \right\}$, then the solution $\{s(t), e(t), i(t), z(t)\}$ of the system (\ref{mainsys12}) is positive for all $t \geq 0$.
\end{Theorem}

\begin{proof}
Taking the first equation of (\ref{mainsys12}) we have

\begin{equation*}
\begin{array}{rl}
\dfrac{ds}{dt} & = \pi - \mu s - \beta si - b sz  \Longrightarrow \dfrac{ds}{dt}  \geq  - \mu s \\
 & \Longrightarrow \dfrac{ds}{s}  \geq - \mu dt   \Longrightarrow      \displaystyle  \int \dfrac{ds}{s}   \geq  \int -\mu dt \\
 & \Longrightarrow s(t)  \geq s(0) e^{-\mu t} \geq  0.
  \end{array}
\label{mainsys4}
\end{equation*} 

From the second equation of (\ref{mainsys12})

\begin{equation*}
\begin{array}{rl}
\dfrac{de}{dt} & = (1-p)\beta si +  (1-l) b sz - \rho ei - \varepsilon e - \mu e  \Longrightarrow \dfrac{de}{dt}  \geq -  (\varepsilon  + \mu )e \\
 & \Longrightarrow \dfrac{de}{e}  \geq -  (\varepsilon  + \mu )dt  \Longrightarrow    \displaystyle  \int \dfrac{de}{e}  \geq   \int -  (\varepsilon  + \mu )dt \\
 & \Longrightarrow e(t)  \geq e(0) e^{-(\varepsilon  + \mu) t} \geq  0.
  \end{array}
\label{mainsys4}
\end{equation*}

From third equation  of (\ref{mainsys12})

\begin{equation*}
\begin{array}{rl}
\dfrac{di}{dt} & =p\beta si  + \rho ei + \varepsilon e  - \delta i  -   \lambda i z - \mu i   \Longrightarrow \dfrac{di}{dt}  \geq -  (\delta + \mu )i \\
 & \Longrightarrow \dfrac{di}{i}   \geq -  (\delta + \mu )dt  \Longrightarrow     \displaystyle \int \dfrac{di}{i}  \geq   \int -  (\delta + \mu  )dt \\
 & \Longrightarrow i(t)  \geq i(0) e^{-(\delta  + \mu) t} \geq  0.
  \end{array}
\label{mainsys4}
\end{equation*} 
\end{proof}

From fourth equation  of (\ref{mainsys12})

\begin{equation*}
\begin{array}{rl}
\dfrac{dz}{dt} & = lbsz + \delta i  +  i (\gamma i + \lambda z) - \mu z  \Longrightarrow \dfrac{ds}{dt}  \geq - \mu z \\
 & \Longrightarrow \dfrac{dz}{z}  \geq  - \mu dt   \Longrightarrow   \displaystyle   \int \dfrac{dz}{z}   \geq  \int -\mu dt \\
 & \Longrightarrow z(t)  \geq z(0) e^{-\mu t} \geq  0.
  \end{array}
\label{mainsys4}
\end{equation*}

\subsection{Existence of the solution}
\begin{Theorem}
The region $D= \{ (s,e,i,z) \in \mathbb {R}^{4}_{+} : \ s+ e+ i + z \leq \dfrac{\pi}{\mu} \}$ is 
 positively invariant and attract all solutions in $\mathbb {R}^{4}_{+}$
 \end{Theorem}
 \begin{proof}
Adding all the equations from (\ref{mainsys}) , gives the rate of change of the total human population 
$\dfrac{dN}{dt}  = \dfrac{dS}{dt}  + \dfrac{dE}{dt} + \dfrac{dI}{dt} + \dfrac{dZ}{dt} $

\begin{equation*}
\begin{array}{rl}
\dfrac{dN}{dt}  &= \pi N - \mu S -   \mu E  - \mu I -  \mu Z \\
\\
 \dfrac{dN}{dt} & =  \pi N - \mu (S + E  + I + Z) \\
\\
  \dfrac{dN}{dt} & =  \pi N - \mu N 
\end{array}
\end{equation*} 

Since $ \dfrac{dN}{dt} =  \pi N - \mu N $ whenever $N(t) >  \pi$ , then $  \dfrac{dN}{dt} < 0$, implying $  \dfrac{dN}{dt}$ is
   bounded by $ \pi N - \mu N$.
Thus, a standard comparison theorem by \cite{{Lakshimikanthan}} can be used to show that : 

\begin{equation*}
N(t) \leq N(0) e^{-\mu t }  + \dfrac{\pi}{\mu} (1- e^{-\mu t })
\end{equation*} 

In particular,  $N(t) \leq  \dfrac{\pi}{\mu} $ if  $ N(0) \leq  \dfrac{\pi}{\mu} $ . Thus $D$ is positively invariant Therefore, the model is epidemiologically and mathematically well posed within the region. 

\end{proof}

\subsection{Basic Reproduction Number}
In this section, we obtained the threshold parameter that governs the spread of rumor like a disease which is called the basic reproduction number which is determined. To obtain the basic reproduction number, we used the next-generation matrix method so that it is the spectral radius of the next-generation matrix. 
The basic reproduction number $R_0$ is an important parameter to characterize the transmission of rumor. In this section, we discuss the existence and uniqueness of Rumor Free Equilibrium (RFE) of the model and its analysis. The model Equations (\ref{mainsys12}) has an RFE given by on a simple calculation: 

\begin{equation*}
M_{0} = \left(S^{*}, E^{*},I^{*}, Z^{*} \right)=  (\dfrac{\pi}{\mu}, \ 0 , \ 0 , \ 0).
\label{mainsys3}
\end{equation*}

The local stability of RFE given will be investigated using the next generation matrix method \cite{Ovaskainen}. We calculate the next generation matrix for the system of the question  (\ref{mainsys12})  by enumerating the number of ways that: 
\begin{itemize}
\item  new spreaders arise 
\item number of ways that individuals can move but only one way to create a spreader. 
\end{itemize}
\begin{center}
    \begin{figure}[ht!]
        \centering
        \includegraphics[width=5cm,height=5cm]{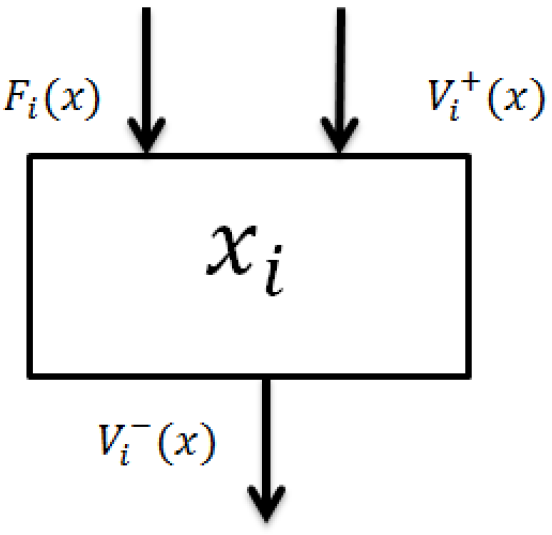}
        \caption{The entry and exit balance sheet}
 %       \label{fig:my_label}
    \end{figure}
\end{center}
Only the equations concerning contaminated and/or contagious individuals (disseminated information) are necessary. 
\begin{equation*}
\left\{
\begin{array}{rl}
%\dfrac{ds}{dt} & = \pi - \mu s - \beta si - b sz \\ \\
\dfrac{de}{dt} & = (1-p)\beta si +  (1-l) b sz - \rho ei - \varepsilon e - \mu e \\ \\
\dfrac{di}{dt} & =p\beta si  + \rho ei + \varepsilon e  - \delta i  -   \lambda i z - \mu i \\ \\
%\dfrac{dz}{dt} & =lbsz\right) + \delta i  +  \lambda iz - \mu z
\end{array}
\right.
\label{mainsys13}
\end{equation*}

We take stock of what goes in and what goes out of each compartment:
\begin{enumerate}
    \item We note $\mathcal{F}_{i}(x)$ rate at which {\bf new spreaders} enter compartment $i$.
    \item We note $\mathcal{V}_{i}^{+}(x)$  those which come from the other compartments by any other cause (displacement, healing, etc...).
    \item We note $\mathcal{V}_{i}^{-}(x)$ the speed of those leaving compartment  $i$ (for example, mortality, movement, change in epidemiological status, ...).
\end{enumerate}

We finally have : 

$$\dot{x}=\mathcal{F}_{i}(x) + \mathcal{V}_{i}(x); \; \; \; \; avec \; \; \; \;  \mathcal{V}_{i}(x) = \mathcal{V}_{i}^{+}(x) + \mathcal{V}_{i}^{-}(x)$$

We denote by $X_S$ the state without disease:
$$X_S = \{ x \in \mathbb{R}^p \mid x_i = 0, \ i = 1, \ldots, p \}$$

The following assumptions are made:
\begin{enumerate}
    \item   $x \geq 0$, \; \; $\mathcal{F}_{i}(x) \geq 0,$ \; \;  $\mathcal{V}_{i}^{+}(x) \geq 0,$ \; \;  $\mathcal{V}_{i}^{-}(x) \geq 0$
    \item If $x_i = 0$, then  $\mathcal{V}_{i}^{-}(x) = 0$. 
    
If there is nothing in a compartment, nothing can  nothing can come out of it. This is the essential property of a compartmental model.
 
    \item For $i \ge p$, then $\mathcal{F}_{i}(x) = 0$. Compartments with an index less than p are "uninfected". By definition, "infected" cannot appear in these compartments.
    \item If $x \in X_S$, then $\mathcal{F}_{i}(x) = 0$ and $\mathcal{V}_{i}^{+}(x) = 0$ for $i = 1,...,p$. If there are no germ carriers in the population, new "infected" cannot appear.
\end{enumerate}
The Jacobian of $f$ is written around the equilibrium point $(f(t, \bar{x}) = 0)$ without disease   $x^{*}$:

$J(x^{*}) = D \mathcal{F}(x^{*}) + D \mathcal{V}(x^{*})$

For $ F = \left [ \dfrac{\partial \mathcal{F}_{i} }{\partial x_j} \right ]_{1 \leq i,j \leq p} $ \; \; et \; \;   $ V = \left [ \dfrac{\partial \mathcal{V}_{i}}{\partial x_j} \right ]_{1 \leq i,j \leq p}$

Where
\begin{enumerate}
    \item $F \geq 0$ is a positive definite matrix and
     \item    $V$  is a Metzler matrix (i,e off-diagonal terms are positive), 
\end{enumerate}

We define $R_0$ then as follows :

$R_0 = \rho(FV^{-1})=det(FV^{-1} -\lambda I)$, Where $\rho$ the spectral ray.

the matrices $F$ and $V$ are defined as follows, respectively: :

\begin{center}
$\mathcal{F} = 
\begin{pmatrix}
(1-p)\beta s \frac{i}{N}   \\
 p\beta s \frac{i}{N}
\end{pmatrix}$
\end{center}

\begin{center}
 $\mathcal{V^{}} = \mathcal{V^{+}} + \mathcal{V^{-}} =
\begin{pmatrix}
- \rho e i - \varepsilon e - \mu e \\
 - \delta i - \mu i + \varepsilon e - \lambda i z 

\end{pmatrix}.$
\end{center}

So, let

F = rate of appearance of new spreaders into the compartment and, \\
V = rate of transfer into (out) of compartment

$$
F=
\begin{pmatrix}
 \begin {array}{cc} 0& \left( 1-p \right) \beta\,S_{{0}}
\\ \noalign{\medskip}0&p\beta\,S_{{0}}\end {array} 

\end{pmatrix}, \quad V=\begin{pmatrix}
\begin {array}{cc} \varepsilon +\mu&0\\ \noalign{\medskip}-
\varepsilon&\delta+\mu\end {array} 
\end{pmatrix}$$

$$
F=
\begin{pmatrix}
 \begin {array}{cc} 0& \left( 1-p \right) \beta\,\dfrac{\pi}{\mu}
\\ \noalign{\medskip}0&p\,\dfrac{\pi \beta\ }{\mu}\end {array} 

\end{pmatrix}, \quad V=\begin{pmatrix}
\begin {array}{cc} \varepsilon +\mu&0\\ \noalign{\medskip}-
\varepsilon&\delta+\mu\end {array} 
\end{pmatrix}$$
$$detV= (\varepsilon + \mu)(\delta + \mu)$$

Hence the next generation matrix with large domain is two dimensional and is given by $FV^{-1}$
\begin{equation}
K=FV^{-1}=
\begin{pmatrix}
 \begin {array}{cc} {\dfrac { \left( 1-p \right) \beta\,\varepsilon
\,\pi}{\mu\, \left( \varepsilon+\mu \right)  \left( \delta+\mu \right) }}
&{\dfrac { \left( 1-p \right) \beta\,\pi}{\mu\, \left( \delta+\mu
 \right) }}\\ \noalign{\medskip}{\dfrac {p\beta\,\varepsilon\,\pi}{\mu\,
 \left( \varepsilon+\mu \right)  \left( \delta+\mu \right) }}&{\dfrac {p
\beta\,\pi}{\mu\, \left( \delta+\mu \right) }}\end {array}
\end{pmatrix}
\label{mainsys4}
\end{equation}
Entry $K_{ij}$ represents expected number of secondary cases in compartment $i$ by an individual in compartment $j$

The dominant eigenvalue of (\ref{mainsys4}) is equal to $R_{0}$ , therefore we evaluate the characteristic equation of (\ref{mainsys4})  by using $det ( FV^{-1} - \lambda I ) =0$ ,which gives after some calculs $$\lambda ^{2} -  \lambda \left [  \dfrac{(1-p)\beta \varepsilon \pi}{\mu (\varepsilon + \mu)(\delta + \mu)} + \dfrac{\pi p \beta (\varepsilon + \mu)}{\mu (\varepsilon + \mu)(\delta + \mu)} \right] =0$$
Finally we have the following expresion of  $R_{0}$ like that : 
$$R_{0}= \dfrac { \beta  \,\pi  \,  \left( \varepsilon + p \mu \right)  }{\mu\, \left( \varepsilon+\mu \right)  \left( \delta+\mu \right) } $$

\begin{thm}
The system of equation (\ref{mainsys12}) is locally asymptotically stable if all its eigenvalues are less than zero at rumor free equilibrium $M_{0} =  (\dfrac{\pi}{\mu}, \ 0 , \ 0 , \ 0)$
\end{thm}

Next generation operator $(FV^{-1})$ gives rate at which individuals in compartment $j$ generate new infections in compartment $i$ times average length of time individual spends in single visit to compartment $j$

\begin{proof}
At rumor-free equilibrium point, the Jacobian  matrix is :
\begin{equation}
J(M_0)=
\begin{pmatrix}
-\mu & 0 & -\dfrac{\beta\pi}{\mu} & -\dfrac{b\pi}{\mu} \\
0 & -\varepsilon-\mu & \dfrac{(1-p)\beta\pi}{\mu} & \dfrac{(1-l)b\pi}{\mu} \\
0 & \varepsilon & \dfrac{p\beta\pi}{\mu}-\delta-\mu & 0 \\
0 & 0 & \delta & \dfrac{lb\pi}{\mu}-\mu
\end{pmatrix}
\label{mainsys5}
\end{equation}

Now we try to calculate the eigenvalues of (\ref{mainsys5}) by finding the characteristic equation using the formula   $det ( J - \lambda I ) =0$

\begin{equation}
det ( J - \lambda I ) = det
\begin{pmatrix}
-\mu - \lambda & 0 & -\dfrac{\beta\pi}{\mu} & -\dfrac{b\pi}{\mu} \\
0 & -\varepsilon-\mu - \lambda& \dfrac{(1-p)\beta\pi}{\mu} & \dfrac{(1-l)b\pi}{\mu} \\
0 & \varepsilon & \dfrac{p\beta\pi}{\mu}-\delta-\mu - \lambda& 0 \\
0 & 0 & \delta & \dfrac{lb\pi}{\mu}-\mu - \lambda
\end{pmatrix} =0
\label{mainsys6}
\end{equation}

From the Jacobian matrix of (\ref{mainsys6}), we obtained a characteristic polynomial:

\begin{equation}
(- \lambda - \mu ) (\lambda ^{3} + a_{2} \lambda ^{2} + a_{1} \lambda + a_{0} ) = 0,
\label{mainsys61} 
\end{equation}
where

\begin{equation}
\left\{
\begin{array}{rl}
a_{2} & = -  \dfrac { lb \pi }{\mu} + 3 \mu -  \dfrac { p \beta \pi }{\mu} + \delta +  \varepsilon  \\
\\
a_{1} & =  (\varepsilon+ \delta)(1- R_{0}) + (\dfrac { lb \pi }{\mu} - \mu)  (\dfrac { p \beta \pi }{\mu} - \delta -  \varepsilon -2 \mu )       \\
\\
a_{0}  & =   - \dfrac { (1-l)b \pi  \varepsilon \,\delta }{\mu}  (\dfrac { lb \pi }{\mu} - \mu) (\varepsilon+ \delta )(1- R_{0})
  \end{array}
  \right.
\label{mainsys7}
\end{equation} 

From (\ref{mainsys61} ) clearly, we see that : 

\begin{equation}
\begin{array}{rl}
- \lambda - \mu & = 0  \Longrightarrow  \lambda   = - \mu  < 0 
  \end{array}
\label{mainsys7}
\end{equation} 

 \mbox { or } 
 
 \begin{equation}
\begin{array}{rl}
 \lambda ^{3} + a_{2} \lambda ^{2} + a_{1} \lambda + a_{0}  & =   0.
  \end{array}
\label{mainsys71}
\end{equation}

From (\ref{mainsys71}) we applied Routh-Hurwitz criteria. By this criteria, (\ref{mainsys71}) has strictly negative real root if and only if   $a_{2} > 0 $  \ ,  $a_{0}  > 0 $  , and $a_{2} \ast a_{1} >  a_{0}.$

Obviously we see that $a_{2}$  is positive because  $ 3\mu + \delta +  \varepsilon  > \dfrac { \pi }{\mu} (lb+p \beta) $, but $a_{0}$  to be positive $1- R_{0}$ must be positive, which leads to $R_{0} < 1$ because $ \mu > \dfrac { lb\pi }{\mu} $. Therefore, RFE will be locally asymptotically stable if and only if $R_{0} < 1$.

\end{proof}
Thus this theorem  implies that for any given rumor in a population, it can be eliminated when $R_{0} < 1$.

\begin{thm}
When $R_0$ is less than or equal to one, the rumor-free equilibrium  is globally asymptotically stable.
\end{thm}

\begin{proof}
A suitable Lyapunov function $L$ to establish the global stability of the rumor-free equilibrium is defined as $L = wI$.
The derivative of the Lyapunov function with respect to time $t$ is:

\begin{equation*}
\begin{array}{rl}
 \dfrac{dL}{dt} &=  \dfrac{di}{dt}  = w [ \ ( p \beta \dfrac{ \pi}{\mu} ) i - (\delta + \mu) i ] \\
 \\
 &=   w (\delta + \mu)  [ \   \dfrac{ p \beta \pi}{\mu (\delta + \mu)}   - 1  ] i \\
 \\
 &=  w (\delta + \mu)  [ \   \dfrac{ p \beta \pi (\varepsilon + \mu) }{\mu (\delta + \mu) (\varepsilon + \mu) }   - 1  ] i \\
 \\
  & \leq  w (\delta + \mu)  [ \   \dfrac{  \beta \pi (\varepsilon + p \mu) }{\mu (\delta + \mu) (\varepsilon + \mu) }   - 1  ] i \ \ \ \ \mbox{ for } \ p \in [0,1] \\
 \\
 & \leq  w (\delta + \mu)  [ R_{0}  - 1  ] i \ \ \ \ \mbox{ for } \ p \in [0,1] \\

  \end{array}
%\label{mainsys10}
\end{equation*} 

If $R_{0} \leq 1$, then $  \dfrac{dL}{dt} \leq 0$ holds. Furthermore, $ \dfrac{dL}{dt} \leq 0$ if and only if $I = 0$.  Hence, $L$ is Lyapunov function on $D$ and the largest compact invariant set in $ \{(I, M, G, R) \in D, \dfrac{dL}{dt}= 0\}$ is the singleton $(\dfrac{\pi}{\mu}, 0, 0, 0)$.  The global stability follows from LaSalle's \cite{LaSalle}   invariance principle, when $R_{0} \leq 1$. Hence, the disease-free equilibrium is globally asymptotically stable.

\end{proof}

\subsection{ The Endemic Equilibrium }
The endemic equilibrium is denoted by $M^{*} = \left(S^{*}, 0 ,I^{*}, Z^{*} \right)$  and it occurs when the disease persists in the community. To obtain it, we equating all the model equations (\ref{mainsys12}) to zero. Then we obtain for $E^{*}= 0$: 

\begin{equation*}
I^{*} = \dfrac{\pi}{\beta S^{*}} - \dfrac{\mu}{\beta} - \dfrac{pb}{\lambda} S^{*} - \dfrac{\delta}{\lambda \beta} - \dfrac{\mu}{\lambda \beta}, 
\end{equation*}

\begin{equation*}
Z^{*} = \dfrac{p \beta}{\lambda}S^{*} - \dfrac{\delta}{\lambda} - \dfrac{\mu}{\lambda}, 
\end{equation*}

When we substitute these expressions into the last equation of (\ref{mainsys12}), we obtained a characteristic polynomial of susceptible,

\begin{equation}
\begin{array}{lll}
&  lBS ( \dfrac{p \beta}{\lambda}S^{*} - \dfrac{\delta}{\lambda} - \dfrac{\mu}{\lambda}) + \delta (\dfrac{\pi}{\beta S^{*}} - \dfrac{\mu}{\beta} - \dfrac{pb}{\lambda} S^{*} - \dfrac{\delta}{\lambda \beta} - \dfrac{\mu}{\lambda \beta}) +  & \\  & \lambda  ( \dfrac{p \beta}{\lambda}S^{*} - \dfrac{\delta}{\lambda} - \dfrac{\mu}{\lambda}) (\dfrac{\pi}{\beta S^{*}} - \dfrac{\mu}{\beta} - \dfrac{pb}{\lambda} S^{*} - \dfrac{\delta}{\lambda \beta} - \dfrac{\mu}{\lambda \beta})- \mu ( \dfrac{p \beta}{\lambda}S^{*} - \dfrac{\delta}{\lambda} - \dfrac{\mu}{\lambda})=0 . 
\end{array}
\label{endemic1}
\end{equation}

From (\ref{endemic1}) we get the following result:

\begin{equation}
\begin{array}{lll}
&  S^{2} \dfrac{pb\beta}{\lambda} (l-p) + S (-\dfrac{lb\delta}{\lambda} - \dfrac{lbu}{\lambda} -pu -   \dfrac{\delta p}{\lambda}  - \dfrac{pu}{\lambda} + \dfrac{pbu}{\lambda}  - \dfrac{p \beta u}{\lambda} ) + \mu^{2}(\dfrac{1}{\lambda} + \dfrac{1}{\beta}  + \dfrac{1}{\lambda \beta} ) & \\  & +  p \beta \pi - \dfrac{\mu \pi}{\beta S}  = 0, 
\end{array}
\label{endemic2}
\end{equation}

which gives  
%after some calculs around the  susceptible we get the following equation : 
\begin{equation}
AS^{3} + BS^{2} + CS + D =0
\end{equation}
where 
\begin{equation}
%\left\{
\begin{array}{rl}
A & =  \dfrac{pb\beta}{\lambda} (l-p),\\ \\
B & =  - \dfrac{1}{\lambda} (lbS+ lb\mu + \delta p+ p\mu - pb\mu + p\beta \mu + p\lambda \mu),\\ \\
C & = p\beta \pi + \mu^{2} (\dfrac{\beta + \lambda + 1}{\lambda\beta}) + \dfrac{\delta \mu}{\lambda},\\ \\
D & = - \dfrac{\mu \pi}{\beta}.
\end{array}
%\right.
\label{mainsys9}
\end{equation} 

\begin{lem}
An  endemic equilibrium point $M^{*}$ exists and is positive if $R_0 > 1.$
\end{lem} 

\section{The Model with Controls}

 Now, we introduce our controls into system (\ref{mainsys12}). As control measures to fight the spread of rumor, we extend our system by including three kinds of controls $u$, $v$, and $w$. 
 
 \begin{enumerate}
\item The first control $u$ is to tell users that the information or publication is false and contains a malicious rumor. 
\item The second control $v$ is through to reduce the number of susceptible who entering and  spread the rumor. 
\item The last one $w$ is also applied by the sceptik or number of people who will question the fake news or those  deactivates an account after learning that it is fake or aimed at spreading the rumor. 
\end{enumerate}
With the aim of better understanding the effects of any control measure of these strategies, we introduce three new variables: $\pi_i, \ i= 1, 2, 3$. We note that $\pi_i=0$ in the absence of control, and $\pi_i=1$ in the presence of control. 

\begin{equation}
\left\{
\begin{array}{rl}
\dfrac{ds}{dt} & = \pi - \mu s - \beta si - b sz  - \pi_1 u s\\ \\
\dfrac{de}{dt} & = (1-p)\beta si +  (1-l) b sz - \rho ei - \varepsilon e - \mu e  - \pi_2 v e\\ \\
\dfrac{di}{dt} & =p\beta si  + \rho ei + \varepsilon e  - \delta i  -   \lambda i z - \mu i  - \pi_3 w i\\ \\
\dfrac{dz}{dt} & =lbsz + \delta i  +  \lambda iz - \mu z + \pi_1 u s 
\end{array}
\right.
\label{mainsys3}
\end{equation}

\subsection{Optimal Control Problem}

We define the objective functional as follows: 

\begin{equation}
J= \int_{0}^{T}  [ I(t) +  \dfrac{1}{2} Au^{2} (t) +\dfrac{1}{2}v^{2} (t) + \dfrac{1}{2}w^{2}  (t)]  dt
\label{functionnal}
\end{equation}

where $A>0$, $B>0$, and $C>0$  are the cost coefficients:

\begin{equation}
J(u^{*}, v^{*}, w^{*})= \min J(u,v,w) \mbox{ over } \Gamma
\end{equation}

The set of admissible controls is defined as follows  : 
\begin{multline}
\Gamma = \left\{ u,v,w \in L^{1}(0,T) \mbox{ such that } (u(t),v(t),w(t))\ \  \in [0,1] \times [0,1]\times [0,1] \  \forall \ t \in [0,T] \right\}
\end{multline}

%The terms $a_{1}I(t)$  and $\dfrac{a_{5}}{2}u_{1}^{2} (t)$ are rather standard terms in objective functionals with similar goals (see, for example, \cite{Hem, Jung}, or, for a contrasting objective functional  \cite{ Urszula, Suzanne})

\subsection{ Existence of an optimal control solution}

Let us consider an optimal control problem having the form (\ref{functionnal})
We analyze sufficient conditions for the existence of a solution to the optimal control problem (\ref{functionnal}). Using a result in Refs. Fleming and Rishel (\cite{Fleming} ), and Hattaf and Yousfi (\cite{Hattaf}), existence of the optimal control can be obtained.

Let us consider now  $L$ as the function that we integrate  $I(t) +  \dfrac{1}{2} Au^{2} (t) +\dfrac{1}{2}v^{2} (t) + \dfrac{1}{2}w^{2}  (t)$ and we get the following lemma
\begin{lemma}
The integrand $L( S , E , I, Z,  u, v, w)$ in the objective functional is convex on $\Gamma$ and there exist constants $c_1$ and $c_2$ such that 
$L(S , E , I, Z,  u, v, w) \geq     c_1 + c_2(|u|^2 + |v|^2 +|w|^2 )^{\frac{\alpha}{2}}$
\end{lemma}

We have the following theorem:
\begin{thm}
Consider the control problem with system (\ref{mainsys3}). There exists an optimal control $(u^{*}, v^{*}, w^{*})\in  \Gamma^{3}$ such
that the control set $\Gamma$ is convex and closed.
\end{thm}

\begin{proof}
The existence of the optimal control can be obtained using a result by Fleming and Rishel (\cite{Fleming} ), checking the following step:
\begin{itemize}
\item  By definition, $\Gamma$ is closed. Take any control $u_1, u_2 \in \Gamma$ and $\lambda \in [0, 1]$. Then $\lambda u_1+ (1-  \lambda)u_2 \geq 0.$. Additionally, we observe that $\lambda u_1  \leq \lambda (1-  \lambda)u_2 \leq (1- \lambda)$; then $ \lambda u_1+  (1-  \lambda)u_2 \leq  \lambda+(1- \lambda)=1$ and 
$0 \leq \lambda u_1 + (1-  \lambda)u_2 \leq 1$, for all $u_1,u_2 \in  \Gamma$ and $\lambda \in [0,1]$. \\ Therefore, $ \Gamma$ is convex and condition 1 is satisfied.

\item The integrand in the objective functional (\ref{functionnal}) is convex on $ \Gamma$. It rests to show that there exist constants $c_1, c_2 > 0$, and $\alpha > 1$ such that the integrand $L(S , E , I, Z,  u, v, w)$ of the objective functional satisfies :

$L(S , E , I, Z,  u, v, w) =  I(t) +  \dfrac{1}{2} Au^{2} (t) +\dfrac{1}{2}v^{2} (t) + \dfrac{1}{2}w^{2}(t) 
    \geq     c_1 + c_2(|u|^2 + |v|^2 +|w|^2 )^{\frac{\alpha}{2}}$
       
The state variables are bounded; let $c_1 =I$,
$c_2 = \inf (\frac{A}{2}, \frac{B}{2}, \frac{C}{2})$, and \ $\alpha = 2$, then it follows that   : 
\begin{equation} 
L(S , E , I, Z,  u, v, w) \geq     c_1 + c_2(|u|^2 + |v|^2 +|w|^2 )^{\frac{\alpha}{2}}
\end{equation} 
Then, from Fleming and Rishel \cite{Fleming} , we conclude that there exists an optimal control.
\end{itemize}
\end{proof}

\subsection{ Characterization of optimal controls}
Let us consider an optimal control problem having the form  (\ref{functionnal}). Pontryagin's Maximum principle \cite{Pontryagin} allows to use costate functions to transform the optimization problem to the problem of determining the pointwise minimum relative to $u^{*}$ ,\  $v^{*}$, and $w ^{*}$ of the Hamiltonian. The Hamiltonian is built from the cost functional (\ref{functionnal})and the controlling dynamics (2) derive the optimality conditions:

\begin{equation}
H = I(t) + \dfrac{1}{2} Au^2 + \dfrac{1}{2} Bv^2 + \dfrac{1}{2} Cw^2 + \sum_{i=1}^{n} p_{i} g_{i}
\end{equation}

where $g_{i}$ denotes the right side of the differential equation of the $i-$th state variables.

\begin{equation}
\begin{array}{lll}
H = &I(t) +  \dfrac{1}{2} Au^{2}  +\dfrac{1}{2}Bv^{2}  + \dfrac{1}{2}Cw^{2}  + p_{1} (\pi - \mu s - \beta si - b sz  - \pi_1 u s) &    \\     & +  \ p_{2} ((1-p)\beta si +  (1-l) b sz - \rho ei - \varepsilon e - \mu e  - \pi_2 v e) & \\
& + \ p_{3} ( p\beta si  + \rho ei + \varepsilon e  - \delta i  -   \lambda i z - \mu i  - \pi_3 w i)  & \\ 
& + \ p_{4} (lbsz + \delta i  +  \lambda iz - \mu z + \pi_1 u s  ) \\
\end{array}
\label{control}
\end{equation}

where the $p_{i}, \ i=1 \,....\, 4$  are the associated adjoints for the states $S , E , I, Z$. The optimality system of equations is found by taking the appropriate partial derivatives of the Hamiltonian (8) with respect to the associated state variable.

The following theorem is a consequence of the maximum principle.

\begin{thm}
Given an  optimal control $(u ^{*}, v^{*}, w ^{*})$ and corresponding solutions to the state system $S^{*} , E^{*} , I^{*}, Z^{*} $  that minimize the objective functional    $ J(u ^{*}, v^{*}, w ^{*})$    there exist adjoint variables  $p_{1}(t), \ p_{2} (t), \  p_{3} (t), \  p_{4} (t)), $ satisfying

\begin{equation}
\left\{
\begin{array}{rl}
\dot{p}_{1} & = -\left[p_{1}(-\mu - \beta i - b z - \pi_1 u) + p_{2}((1-p)\beta i + (1-l) b z) + p_{3}(p\beta i) + p_{4}(lbz + \pi_1 u)\right] \\ \\
\dot{p}_{2} & = -\left[p_{2}(-\rho e - \varepsilon - \mu - \pi_2 v) + p_{3}(\rho i + \varepsilon)\right] \\ \\
\dot{p}_{3} & = -\left[1 + p_{1}(-\beta s) + p_{2}((1-p)\beta s - \rho e) + p_{3}(p\beta s + \rho e - \delta - \lambda z - \mu - \pi_3 w)\right] \\ \\
\dot{p}_{4} & = -\left[p_{1}(- bs) + p_{2}(1-l)bs + p_{3}(\lambda i) + p_{4}(lbs + \lambda i - \mu)\right]
\end{array}
\right.
\label{mainsys9}
\end{equation}

with the transversality conditions 

\begin{equation}
\begin{array}{lll}
p_{1}(T)   & =   & 0     \\  
 
 p_{2}(T)   & =   & 0 \\
 
p_{3}(T)   & =   & 0 \\ 

p_{4}(T)   & =   & 0 \\
\end{array}
\label{control}
\end{equation}

Furthermore, we may characterize the optimal pair by the piecewise continuous functions  and for $\pi_1 = \pi_2 = \pi_3 = 1$

\begin{equation}
\begin{array}{lll}
u^ {*}(t) &=& \min \left\{ \max \left (0, \dfrac{ \pi_1 \, S}{A} (p_{1} - p_{4} )\right ), 1 \right \},\\
\\
v^ {*} (t)&=&  \min \left\{ \max \left (0, \dfrac{\pi_2\, p_{2} \,e}{B} \right ), 1 \right \},\\
\\
w^ {*} (t)&=&  \min \left\{ \max \left (0, \dfrac{ \pi_3\, p_{3} \,i}{C} \right ), 1 \right \},\\
\end{array}
\label{control}
\end{equation}

\label{thm4}
\end{thm}

\begin{proof}
The existence of optimal controls follows from Corollary 4.1 of Fleming and Rishel \cite{Fleming} since the integrand of $J$ is a convex function of $(u, v ,w )$, and the state system satisfies the Lipchitz property with respect to the state variables because the state solutions are $L^{\infty}$ bounded. The following can be derived from  Pontryagin's maximum principle (\cite{Pontryagin}):

\begin{equation*}
\dot p_{1}= - \dfrac{ \partial H}{\partial S}\,,  \dot p_{2}= - \dfrac{ \partial H}{\partial E}\,, \dot p_{3}= - \dfrac{ \partial H}{\partial I}\, \ \dot p_{4}= - \dfrac{ \partial H}{\partial Z} , 
\end{equation*}

with \, $p_{i}(T)=0 , for \ i=1,\ 2,\ 3,\ 4$. evaluated at the optimal controls and the corresponding states, which results in adjoint system of theorem  (\ref{thm4}). The Hamiltonian $H$ is minimized with respect to the controls at the optimal controls; therefore, we differentiate $H$ with respect to $u$ , $v $, and $w$ on the set  $\Gamma$ , respectively, thereby obtaining the following optimality conditions:

\begin{equation*}
\begin{array}{lll}
\dfrac{ \partial H}{\partial u} = A u(t) \, - p_{1}\pi_1 S  \, +  p_{4}\pi_1 S =0   \Longleftrightarrow  u(t)= \dfrac{ \pi_1 \, S}{A} (p_{1} - p_{4} ) \\ \\

\dfrac{ \partial H}{\partial v} = B v(t) \, - p_{2} \pi_2 e =0  \Longleftrightarrow  v(t)=  \dfrac{\pi_2\, p_{2} \,e}{B} \\ \\

\dfrac{ \partial H}{\partial w} = C w(t) \, - p_{3} \pi_3 i =0 \Longleftrightarrow  w(t)= \dfrac{\pi_3 \, p_{3} \, i}{C} \\ 
\end{array}
\end{equation*}\\

Solving for $u^{*}$ , $v^{*} $, and $w^{*}$, we obtain  for the bounds in $\Gamma$ of the controls,

\begin{equation}
\begin{array}{lll}
u^ {*}(t) &=& \min \left\{ \max \left (0, \dfrac{ \pi_1 \, S}{A} (p_{1} - p_{4} )\right ), 1 \right \},\\
\\
v^ {*} (t)&=&  \min \left\{ \max \left (0, \dfrac{\pi_2\, p_{2} \,e}{B} \right ), 1 \right \},\\
\\
w^ {*} (t)&=&  \min \left\{ \max \left (0, \dfrac{ \pi_3\, p_{3} \,i}{C} \right ), 1 \right \},\\
\end{array}
\label{control-proof}
\end{equation}

However, if $\pi_i=0$ where $i =1, 2, 3$  the controls attached to his case will be eliminated and removed.

By the standard variation arguments with the control bounds, we obtain the optimal solutions (\ref{control})
\end{proof}

\subsection{Numerical Simulation}

In this section, we present the results obtained by solving numerically the optimality system.This system consists of the state system, adjoint system, initial and final time conditions, and the control characterization. So, the optimality system is given by the following:

\begin{table}[H]
	\caption{Parameters model formulations and their descriptions}
	\begin{center}
		\begin{tabular}{c|c|c|c|c|c|c|c|c|clcl} 
			\begin{tabularx}{15cm}{|c|X|c|c}
				\hline
				Parameter & Description   &  Units &  Value   \\
				\hline
				$\pi $ & Recruitment rate into the Susceptible population& Per unit time &  Value  \\
				\hline
				$\beta $ & Rate of contact between S and I  & Per unit time  &  Value  \\
				\hline
				$b $ & Rate of contact between S and Z & Per unit time  &  Value \\
				\hline
				$\rho$ & Rate of contact between E and I  & Per unit time &  Value   \\
				\hline
				$\varepsilon$ & Incubation rate  & Per unit time  &  Value \\
				\hline
				$\dfrac{1}{ \varepsilon}$ & Average Incubation rate  & Per unit time  &  Value \\  
				\hline
				$p $ & Transmission rate S->I, given contact with I & Unit-less &  Value \\
				\hline
				$l $ & Transmission rate S->Z, given contact with Z & Unit-less  &  Value \\
				\hline
				$1-l $ & S->E Probability given contact with skeptics  & Unit-less &  Value   \\
				\hline
				$1-p $ & S->E Probability given contact with adopters & Unit-less &  Value  \\
				\hline
				$\mu $ & Deconnect rate of network & Per unit time  &  Value \\
				\hline
			\end{tabularx}
		\end{tabular}
	\end{center}
	\label{tab1}
\end{table}%
In this paragraph, we give numerical simulation to highlight the effectiveness of the strategy that we have developed in the framework of eliminating the rumor and limit its spread; the initial values are the same as in Table 1; with regard to other initial values, they proposed values after a statistical study.

\subsubsection{Numerical Simulation for $R_{0}<1$}
  
  Figure \ref{fig3}  illustrates the dynamics of SEIZ in the absence of controls, and we can see that the starting number is initially low and that the number of susceptible individuals has decreased over time. We note that from the outset, infected and susceptible individuals tend towards $0$ and that only skeptics have increased from their initial state to a higher number. This figure shows that if we have a low density of individuals, then within a few days or months, the false tends to disappear because of the skeptics.
  
 \begin{figure}[H]
	\centering
		\includegraphics[width=0.5\textwidth]{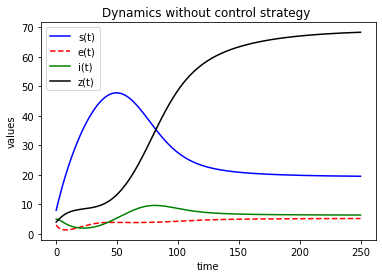}
	\caption{Dynamics of the model with the values $\pi = 10; \beta = 0.007; \mu = 0.5; \varepsilon = 0.06; \delta = 0.05; p = 0.09767; \lambda = 0.0084231; \rho = 0.21431; l = 0.005234; b = 0.00539.$}
			\label{fig3}
 \end{figure}

\begin{figure}[H]
  \begin{minipage}[b]{0.5\linewidth}
   \centering
   \includegraphics[width=8cm,height=6cm]{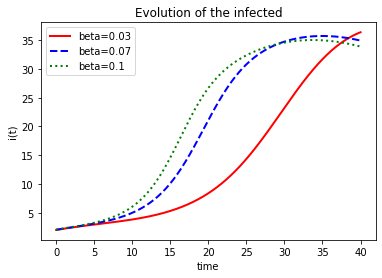}     
  \end{minipage}
\hfill
  \begin{minipage}[b]{0.48\linewidth}
   \centering
   \includegraphics[width=8cm,height=6cm]{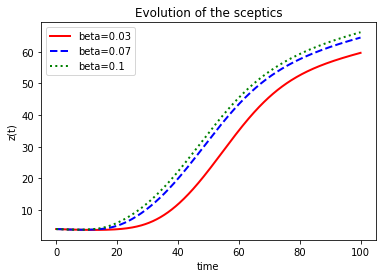}     
  \end{minipage}
  \caption{Dynamics of infected and Sceptiks for different values de $\beta$ and all the other values constants  for  $R_{0}<1$ }
  \label{fig4}
\end{figure}

\subsubsection{Numerical Simulation for $R_{0}>1$}
  
Figure \ref{fig12} represents the dynamics of SEIZ in the absence of controls and we can see that the number of susceptible people has increased from its initial state to a number to stabilize. We note that from the outset, the infected become more and more numerous in sharing false information until they stabilize, while the skeptics remain much lower than the infected. This balance shows that false information is quickly relayed.

 \begin{figure}[H]
	\centering
		\includegraphics[width=0.5\textwidth]{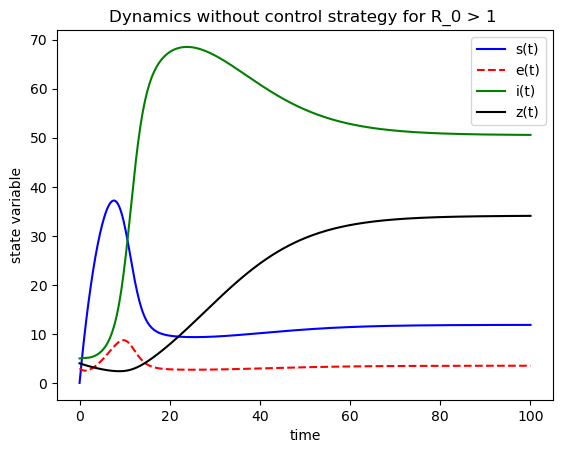}
	\caption{Dynamics of the model with the values $\pi = 50; \beta = 0.07; \mu = 0.5; \varepsilon = 0.06; \delta = 0.05; p = 0.09767; \lambda = 0.0084231; \rho = 0.21431; l = 0.005234; b = 0.00539.$}
			\label{fig12}
 \end{figure}
 
In order to increase the stifling rate A. we should generally improve the level of scientific knowledge of the public in society. This way, the public can clearly identify
general rumors and not easily believe and spread them. Figure  \ref{fig13}  illustrates how
the number of the sceptiks people changes over time with different Rate of contact between S and I noted $\beta$. From   \ref{fig13} , we can establish that
when the Rate of contact between S and I  decreases, the number of sceptiks people decreases.
Therefore reducing this Rate of contact can control the spread of rumors.
\begin{figure}[H]
  \begin{minipage}[b]{0.5\linewidth}
   \centering
   \includegraphics[width=8cm,height=6cm]{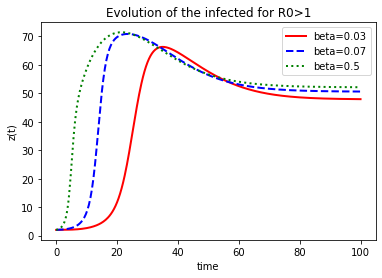}     
  \end{minipage}
\hfill
  \begin{minipage}[b]{0.48\linewidth}
   \centering
   \includegraphics[width=8cm,height=6cm]{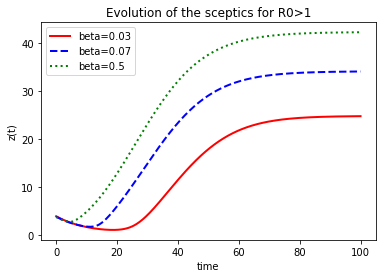}     
  \end{minipage}
  \caption{Dynamics of infected and Sceptiks for different values de $\beta$ and all the other values constants  for  $R_{0}>1$}
  \label{fig13}
\end{figure}

 \subsubsection{Case 1: Applying Only Control $u$}
Since it will be applied to ignorant individuals, we will be limited to displaying and comparing the curves of infected and Sceptiks in  case with control strategy. In this scenario, we simulate the case where we apply a single control u with which we inform a portion of the ignorants by the false information, so we win this proportion in favor of stiflers. We observe from  our Figure that, some days after the implementation of the strategy, the impact will start to appear as we note that the number will gradually decrease until it stabilizes. On the other hand, the number of Sceptiks in this model will suddenly start to rise. This change is probably due to the fact that the control is aimed at telling the ignorant people to turn to stifler ones. In this way, we win a number of people in the fight against the spread of false news.
  
 \begin{figure}[H]
	\centering
		\includegraphics[width=0.7\textwidth]{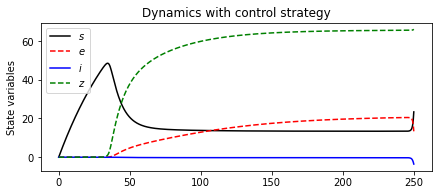}
	\caption{Dynamics of the model with the control $u$}
			\label{fig1}
 \end{figure}
  \begin{figure}[H]
	\centering
		\includegraphics[width=0.6\textwidth]{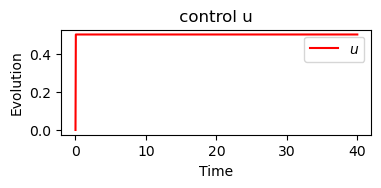}
	\caption{Optimal control u for  SEIR optimal control problem}
		\label{fig1}
 \end{figure}

 \subsubsection{Case 2: Applying Only Control $v$} 
Here, we will implement only control v, that the effect of the strategy will appear  on the number of infected as the number will gradually decrease. This rapid change is attributed to the fact that control directly targets this group. 
 
In the second scenario, we apply a single v-control, but this time one that focuses on broadcasters. The figure shows that the number of diffusers has decreased, but this time we see that the number of sceptoks tends towards zero; however, we note that this control also leads to a loss of both s diffusers and skeptics. 

 \begin{figure}[H]
	\centering
		\includegraphics[width=0.7\textwidth]{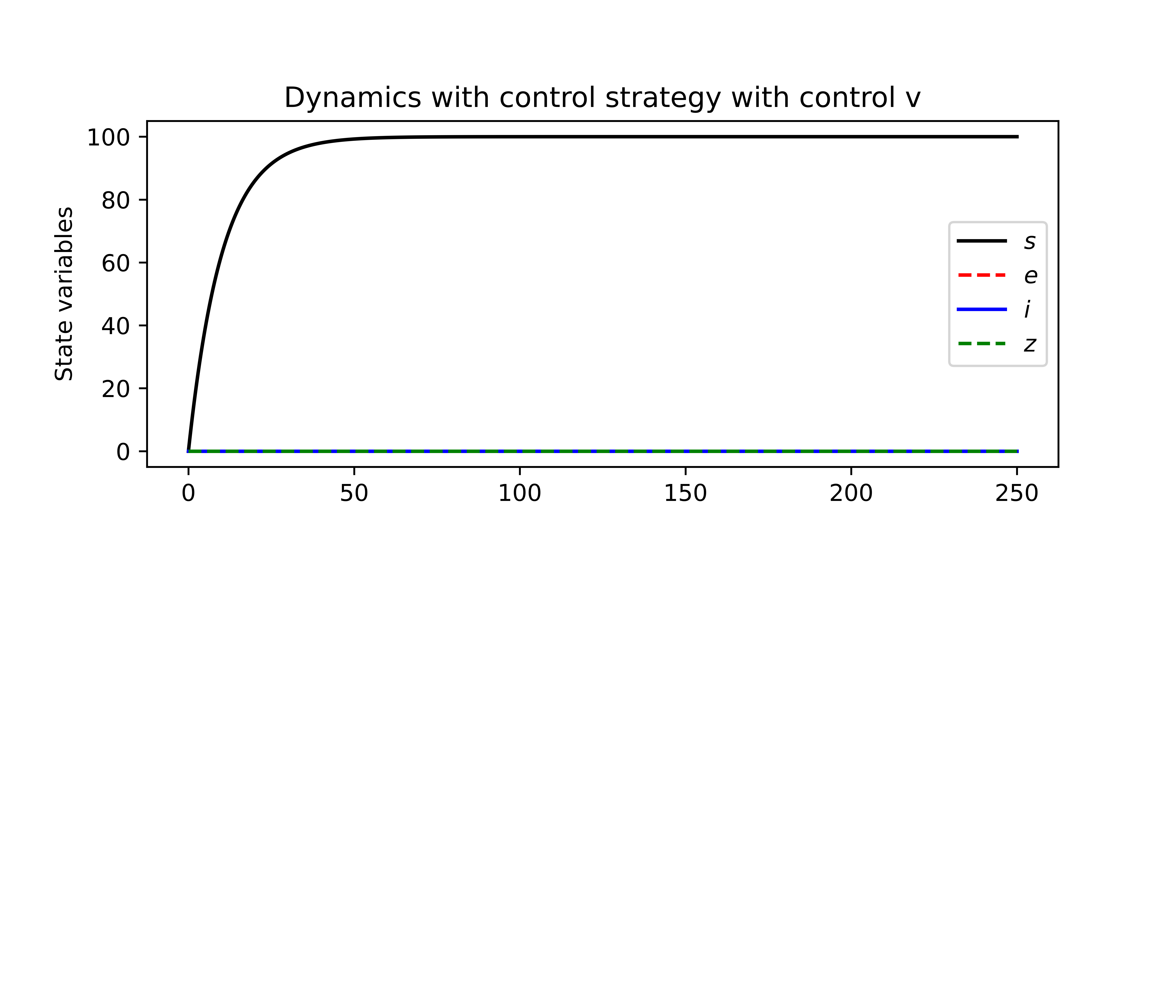}
	\caption{Dynamics of the model with the control $v$}
			\label{fig1}
 \end{figure}

  \begin{figure}[H]
	\centering
		\includegraphics[width=0.7\textwidth]{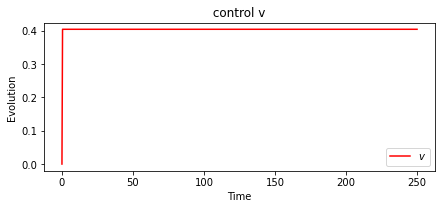}
	\caption{Optimal control v for the SEIR optimal control problem }
			\label{fig1}
 \end{figure}

  \subsubsection{Case 2: Applying Only Control  $u$ , $v$ and $w$} 
In this strategy, we implemented the three controls as an intervention to eradicate the rumor from the community or population. Figure \ref{figcontroluvw} shows that the number of infectious individuals and sceptics is zero for some time before they start to increase progressively, with a clear progression of sceptics above broadcasters or infected individuals. Consequently, the application of this strategy is effective in eradicating rumor as a community disease in a specified period of time.

   \begin{figure}[H]
	\centering
		\includegraphics[width=0.7\textwidth]{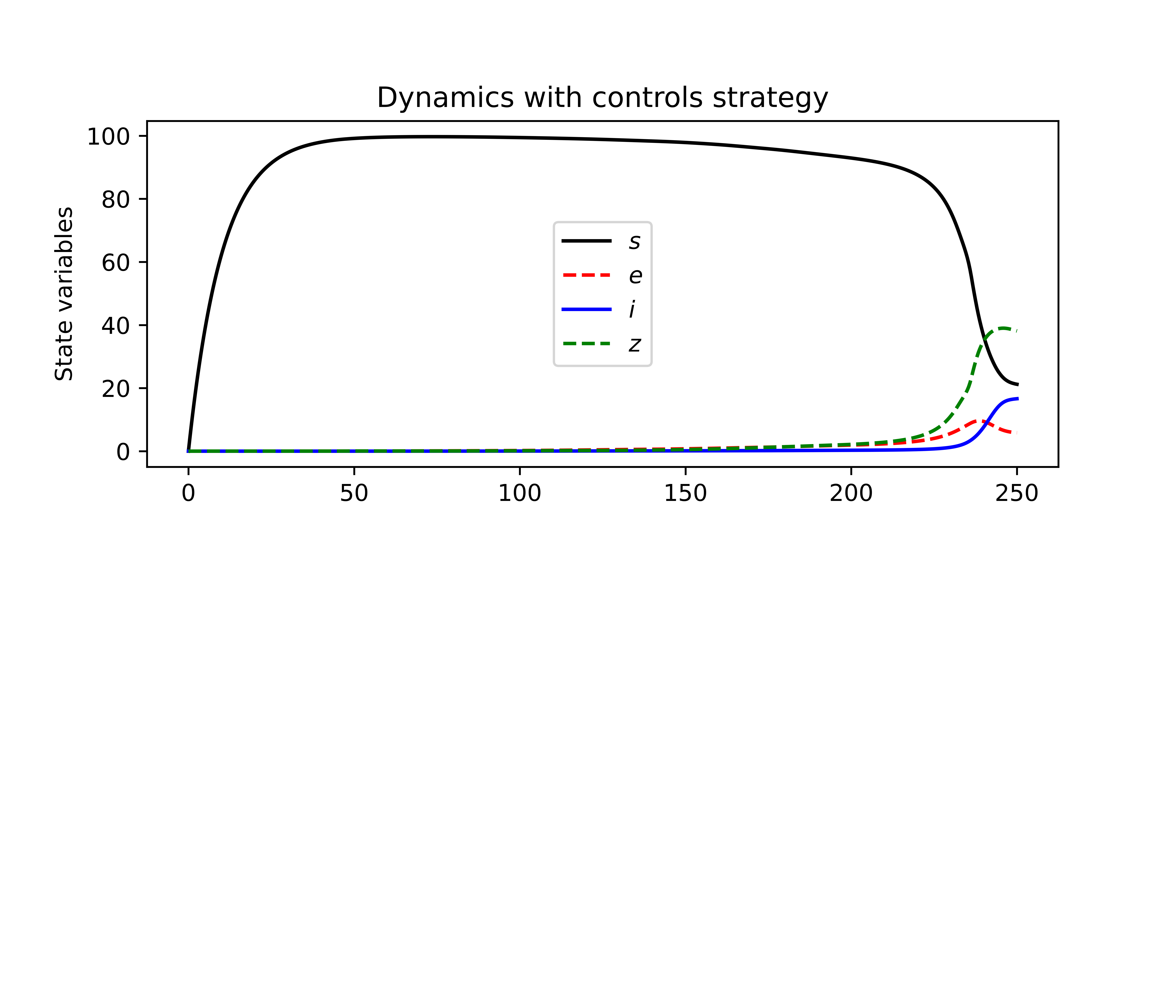}
	\caption{Dynamics of the SEIR model for all optimals controls u , v and w applied }
			\label{figcontroluvw}
 \end{figure}

 \begin{figure}[H]
	\centering
	\subfigure[\label{fig:states-all-sc3}]{\includegraphics[width=0.45\linewidth]{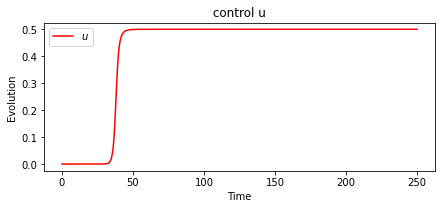}}
	\subfigure[\label{fig:states-all-sc3-1-e3}]{\includegraphics[width=0.45\linewidth]{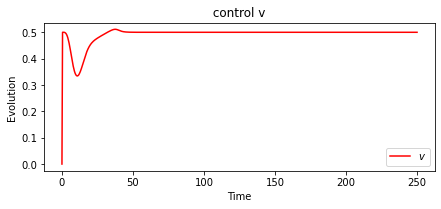}}
	\subfigure[\label{fig:states-all-sc3-3-e2}]{\includegraphics[width=0.45\linewidth]{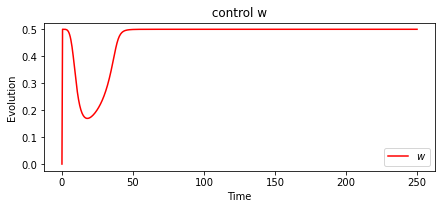}}
	\caption{optimals controls u , v and w of the third strategy. }
	\label{fig:states-sc3}	
\end{figure}

\newpage 	
\section{Conclusions and future Work}
In this paper, we give a new simple mathematical model which describes the dynamics of rumor propagation. The model is based on two compartmental models by combining them in order to take into account more factors that are involved in the dynamic. Three control strategies were introduced, and referring to the introduction of three new variables , $i = 1, 2, 3,$ we could study and combine several scenarios in order to see the impact and the effect of each one of these controls on the reduction of the rumor spread. The goal is achieved and the numerical resolution of the system with difference equations as well as the numerical simulations enabled us to compare and see the difference between each scenario in a concrete way. The purpose of the work is achieved and we have proved the effectiveness of our strategy and its importance in fighting the spread of any rumor  like throughout any social network.


\begin{thebibliography}{99}


      \bibitem{Allport}
        \newblock W. Allport and L. Postman, 
        \newblock \emph{ The Psychology of Rumor}, 
         \newblock  Holt Rinehart  Winston, New York, NY, USA, 1947.

         
      \bibitem{Anderson}    
     \newblock  R.M.  Anderson and  R. May,(1992) 
      \newblock \emph{  Infectious Diseases of Humans: Dynamics and Control of (Paper Backed) Oxford}.
      \newblock  Oxford University Press, New York.
     
          \bibitem{Bartholomew}
     \newblock  R. E. Bartholomew and P. Hassall,,
     \newblock \emph{A Colorful History of Popular Delusions, Prometheus Books, Buffalo,},
     \newblock  NY, USA, 2015.

     
      \bibitem{Bettencourt}
    \newblock   L.M.A. Bettencourt, , A.Cintrón-Arias, , D.I.Kaiser,  C.Castillo-Chávez, 
      \newblock \emph{ The power of a good idea: Quantitative modeling of0the spread of ideas from epidemiological models. } 
        \newblock Phys. A Stat. Mech. its Appl. https://doi.org/10.1016/j.physa.2005.08.083 (2006)
     

     \bibitem{Bordia}
     \newblock  P. Bordia and N. Difonzo,
     \newblock \emph{ Psychological motivations in rumor spread,” Analysis of Commercial Rumors from the Perspective of Marketing Managers: Rumor Prevalence, Effects, and Control Tactics },
     \newblock Progr. Fract. Differ. Appl,  87–101, (2005).
     

     \bibitem{bhattachary2014predator}
     \newblock S. Bhattacharya,M. Martchevaa, X. Z. Lib,
     \newblock \emph{A predator-prey-disease model with immune response in infected prey},
     \newblock J. Math.Anal.Appl. (2014), \textbf{411}, 297-313.

     \bibitem{Chen}
     \newblock   J.Chen; C.Chen,Q.Song, Y. Zhao, L. Deng,  R.Xie,  S.Yang, 
     \newblock \emph{Spread Mechanism and Control Strategies of Rumor Propagation Model Considering Rumor Refutation and Information Feedback in Emergency}
     \newblock Management. Symmetry 2021, 13, 1694. https://doi.org/10.3390/ sym13091694

   \bibitem{Daley}    
         \newblock  D. J. Daley and D. G. Kendall, 
        \newblock \emph{  “Statistics epidemics and rumours,” }
  \newblock    Nature, vol. 204, p. 1964.
     

     \bibitem{dynamic2013Javidi}
     \newblock M. Javidi, N. Nyamoradi,
     \newblock \emph{Dynamic analysis of a fractional order prey-predator interaction with harvesting},
     \newblock Applied Mathematical Modelling (2013), \textbf{37}, 8946-8956.
     
           \bibitem{Estabrook}
       \newblock       G.F. Estabrook,, 2004. 
      \newblock \emph{ Mathematical Biology: I: An Introduction}. 
      \newblock Third Edition. Interdisciplinary Applied Mathematics, Volume 17 . By J D Murray. New York: Springer . 59.95. xxiii + 551 p; ill.; index. ISBN: 0–387– 95223–3. 2002. Mathematical Biology: II: Spatial Models an. Q. Rev. Biol. https://doi.org/10.1086/421587
     
     \bibitem{Fan}
     \newblock  L. Fan, W. Wu, X. Zhai, K. Xing, W. Lee, and D.-Z. Du,,
     \newblock \emph{Maximizing rumor containment in social networks with constrained time},
     \newblock Social Network Analysis and Mining, (2014), \textbf{13}, 171-175.


     \bibitem{Fleming} 
      \newblock  W.H. Fleming,  R.W,Rishel. 
   \newblock \emph{ Deterministic and stochastic optimal contro.} 
          \newblock Springer, New York (1975)

     \bibitem{Fromm}
     \newblock J. Fromm, S. Melzer, B. Ross, and S. Stieglitz,
     \newblock \emph{Trump versus Clinton-Twitter Communication during the US Primaries},
     \newblock Palgrave Macmillan, London, UK, 2016.



      \bibitem{Ghazzali}    
      \newblock  R. Ghazzali, A. E. A. Laaroussi, A. EL Bhih, and M. Rachik, 
      \newblock \emph{ On the control of a reaction-diffusion system: a class of SIR distributed parameter systems,” International Journal of
Dynamics and Control}, 
 \newblock  vol. 7, no. 3, pp. 1021–1034, 2019.

 
     \bibitem{Goffman}
     \newblock W. Goffman and V. A. Newill, 
     \newblock \emph{Generalization of epidemic theory: an application to the transmission of ideas},
     \newblock Nature, vol. 204, no. 4955, pp. 225–228, 1964..
     

\bibitem{Hattaf} 
     \newblock  K,Hattaf,  N. Yousfi (2011) .
    \newblock \emph{Dynamics of HIV infection model with therapy and cure rate}. 
     \newblock Int J Tomogr Stat 16(11):74–80


     
     \bibitem{heggeru2013local}
     \newblock C. M.Heggeru, K. Lan,
     \newblock \emph{Local stability analysis of ratio-dependent predator–prey models with predator harvesting rates},
     \newblock Applied Mathematics and Computation (2015), \textbf{270}, 349-357.
     


     \bibitem{huanga2014bif}
     \newblock J. Huanga, S. Ruan, J. Song,
     \newblock \emph{Bifurcations in a predator–prey system of Leslie type with generalized Holling typeIII functional response},
     \newblock J. DifferentialEquations (2014), \textbf{257}, 1721-1752.
     

          
     \bibitem{Isea}    
     \newblock    R. Isea, K. E A Lonngren. 
          \newblock \emph{ A New Variant of the SEIZ Model to Describe the Spreading of a Rumor. }        \newblock International Journal of Data Science and Analysis, 2017, 3 (4), pp.28-33. ff10.11648/j.ijdsa.20170304.12ff. ffhal-03521046f


     \bibitem{lv2013prey}
     \newblock Y. lv, R. Yuan, Y. Pei,
     \newblock \emph{A prey-predator model with harvesting for fishery resource with reserve area},
     \newblock Applied mathematical Modelling (2013), \textbf{37}, 3048-3062.

     
     \bibitem{Jain}
     \newblock  A. Jain, V. Borkar, and D. Garg,,
     \newblock \emph{Fast rumor source identification via random walks},
     \newblock   Social Network Analysis and Mining  \textbf{vol. 6, no. 1, p. 62, 2016.}

     
     \bibitem{Jia}
     \newblock J. Jia and W. Wu,
     \newblock \emph{A rumor transmission model with incubation in social networks},
     \newblock Physica A: Statistical Mechanics and its Applications, vol. 491, pp. 453–462, 2018.
     

     \bibitem{kar2003selective}
     \newblock T. K. Kar,
     \newblock \emph{Selective harvesting in a predator-prey fishery with time delay },
     \newblock Mathematical and Computer Modelling (2003), \textbf{38}, 449-458.
     

     \bibitem{Kermack}
     \newblock W.O.Kermack and A.G.McKendrick,
     \newblock \emph{A contribution to the mathematical theory of epidemics},
     \newblock Proceedings of the Royal Society of London. Series A, vol. 115, pp. 700–721, 1927.

      

     \bibitem{Laarabi}
     \newblock H. Laarabi, A. Abta, M. Rachik, and J. Bouyaghroumni,
     \newblock \emph{Stability analysis of a delayed rumor propagation model},
     \newblock Differential Equations and Dynamical Systems, vol. 24, no. 4, pp. 407–415, 2016..
     
     
     \bibitem{Lakshimikanthan}
      \newblock  V.Lakshimikanthan,  S. Leela, and A.A, Martynyuk
       \newblock \emph{ Stability of Analysis of Nonlinear System}. 
      \newblock     Mareel Dekker, New York. (1989)

     
    \bibitem{LaSalle}  
      \newblock    J. P. LaSalle, 
     \newblock \emph{ The Stability of Dynamical Systems}, 
       \newblock  vol. 25. Philadelphia, PA, USA: SIAM, 1976.
     
          
     \bibitem{ Lee}
     \newblock C. Lee and K. Quealy,
     \newblock \emph{The 337 People, Places and Things Donald
Trump Has Insulted on Twitter:A Complete List, 2016},
     \newblock https:// www.nytimes.com/interactive/2016/01/28/upshot/donald-trump- twitter-insults.html..
     
     
     \bibitem{Maki}
     \newblock  D. P. Maki and M. Thompson,
     \newblock \emph{Mathematical Models and Applications},
     \newblock  Prentice-Hall, Englewood cliffs, NJ, USA, 1974.

     \bibitem{mansal2014math}
     \newblock F. Mansal, T. N. H. P. Auger, M. Balde,
     \newblock \emph{A Mathematical Model of a Fishery with Variable Market Price: Sustainable Fishery/Over-exploitation},
     \newblock Acta Biotheor (2014), \textbf{62}, 305-323.

     

      \bibitem{Ndii}   
       \newblock M. Z. Ndii, E. Carnia, and A. K. Supriatna,   
       \newblock \emph{ “Mathematical models for the spread of rumors: a review,” in Proceedings of the 6th International Congress on Interdisciplinary }\newblock Behavior and Social Sciences (ICIBSoS 2017), Bali, Indonesia, July 2018.

      \bibitem{Ovaskainen}    
       \newblock  O. Ovaskainen and B. Meerson, 
       \newblock \emph{Stochastic models of population extinction,}
     \newblock    Trends in Ecology  Evolution, vol. 25, no. 11,pp. 643–652, 2010, http://linkinghub.elsevier.com/retrieve/pii/ S0169534710001801.




          
     \bibitem{Pontryagin}
     \newblock L. S. Pontryagin, V. G. Boltyanskii, R. V. Gamkrelidze,E. Mishchenko,
     \newblock \emph{The Mathematical Theory of Optimal Processes (International Series of Monographs in Pure and Applied Mathematics)},
     \newblock Interscience, New York, NY, USA , 1962

     
      \bibitem{Rapoport}  
     \newblock      A.Rapoport and L.I.Rebhun,
         \newblock \emph{    “On the mathematical theory of rumor spread,” }
        \newblock  The Bulletin of Mathematical Biophysics, vol. 14, no. 4, pp. 375–383, 1952.
 
     \bibitem{rebaza2012dyanamic}
     \newblock J. Rebaza,
     \newblock \emph{Dynamics of prey threshold harvesting and refuge},
     \newblock Journal of Computational and Applied Mathematics (2012), \textbf{236}, 1743-1752.

     \bibitem{disease2012sahoo}
     \newblock B. Sahoo, S. Pori,
     \newblock \emph{Disease control in a food chain model supplying alternative food},
     \newblock Applied Mathematical Modelling (2013), \textbf{37}, 5653-5663.

   
     
     \bibitem{Sabato}
     \newblock  L. J. Sabato, K. Kondik, and G. Skelley,
     \newblock \emph{Republicans 2016: What to
Do with the Donald?, 2015, http://centerforpolitics.org/crystalball/
articles/republicans-2016-what-to-do-withthe-donald/.}, 

     
     \bibitem{Santhoshkumar}
     \newblock S. Santhoshkumar and L. D. Dhinesh Babu,  
     \newblock \emph{Earlier detection of rumors in online social networks using certainty-factor- based convolutional neural networks},
     \newblock Social Network Analysis and Mining, \textbf{vol. 10, no. 1, p. 20, 2020},
     
        

     

     \bibitem{Zhang}
     \newblock R. Zhang and D. Li,
     \newblock \emph{Rumor propagation on networks with community structure},
     \newblock  Physica A: Statistical Mechanics and its Applications, vol. 483, pp. 375–385, 2017.

     \bibitem{Zheng}
     \newblock  M. H. Zheng, L. Y. Lv, and M. Zhao,
     \newblock \emph{Spreading in online social networks: the role of social reinforcement},
     \newblock  Physical Review E, vol. 88, no. 1, Article ID 012818, 2013.     

     \bibitem{zephoria}
     \newblock https://zephoria.com/top-15-valuable-facebook-statistics/.

\end{thebibliography}
  \end{document}